\documentclass[11pt]{article}
\usepackage{textcomp}
\usepackage[a4paper, margin=1.25in]{geometry}
\usepackage[T1]{fontenc} 
\usepackage[normalem]{ulem} 

\usepackage[dvipsnames,svgnames,x11names]{xcolor}

\colorlet{RED}{red} 

\colorlet{WHITE}{white}
\colorlet{FORESTGREEN}{ForestGreen}
\colorlet{MAROON}{Maroon}

\usepackage[nocompress]{cite}

\usepackage{balance}

\usepackage{listings}

\lstdefinelanguage{XML}
{
basicstyle=\ttfamily\footnotesize,
  morestring=[b]",
  moredelim=[s][\bfseries\color{Maroon}]{<}{\ },
  moredelim=[s][\bfseries\color{Maroon}]{</}{>},
  moredelim=[l][\bfseries\color{Maroon}]{/>},
  moredelim=[l][\bfseries\color{Maroon}]{>},
  morecomment=[s]{<?}{?>},
  morecomment=[s]{<!--}{-->},
  commentstyle=\color{gray},
  stringstyle=\color{blue},
  identifierstyle=\color{red}
}

\usepackage{moreverb}

\usepackage[nounderscore]{syntax}

\usepackage[pdftex]{graphicx}
\graphicspath{{./figures/}}

\DeclareGraphicsExtensions{.pdf}

\usepackage[cmex10]{amsmath}
\usepackage{amssymb}
\usepackage{mathtools}
\usepackage{amsthm}
\usepackage{amsfonts}
\usepackage{gensymb}

\newtheorem{theorem}{Theorem}
\newtheorem{lemma}[theorem]{Lemma}

\usepackage{subfig}

\usepackage{algorithm}
\usepackage{algorithmic}

\usepackage{multirow} 
\usepackage{rotating} 
\usepackage{booktabs} 
\usepackage{colortbl} 
\usepackage{tablefootnote} 
\usepackage{array}
\newcolumntype{L}[1]{>{\raggedright\let\newline\\\arraybackslash\hspace{0pt}}m{#1}}
\newcolumntype{C}[1]{>{\centering\let\newline\\\arraybackslash\hspace{0pt}}m{#1}}
\newcolumntype{R}[1]{>{\raggedleft\let\newline\\\arraybackslash\hspace{0pt}}m{#1}}

\usepackage[pdftex,colorlinks=true,urlcolor=blue,citecolor=blue]{hyperref}

\usepackage{xspace}

\usepackage{enumitem}

\newcommand{\etal}{et~al.\@}

\usepackage{blindtext}

\begin{document}

\title{\huge Folding-Free Zero-Noise Extrapolation by Layout-induced Noise Diversity}

\author{
Debarthi Pal and Yogesh Simmhan
\\~\\
\textit{Indian Institute of Science, Bangalore}\\
\texttt{Email: \{debarthipal, simmhan\}@iisc.ac.in}\\
}

\date{}

\maketitle

\begin{abstract}
    Near-term quantum processors operate in a noise‑dominated regime, motivating error‑mitigation techniques that recover accurate expectation values without full fault tolerance. Zero‑Noise Extrapolation (ZNE) is a biased error‑mitigation method that does not provide any rigorous error bound. However, it is still among the most widely used approaches due to its simplicity. Nevertheless its effective application requires nontrivial technical choices, most notably the selection of noise-scaling factors and extrapolation models, making ZNE sensitive to user expertise and often necessitating costly trial-and-error procedures. Here, we introduce Folding-Free Zero-Noise Extrapolation (FF-ZNE), a method that removes the need for noise-factor selection by achieving effective noise amplification without circuit folding. FF-ZNE exploits isomorphic hardware layouts with distinct native noise profiles, such that executing a fixed circuit across these layouts induces controllable variations in the effective noise strength. Under a depolarizing noise model, we analytically show that the resulting extrapolation admits a fixed linear form, eliminating extrapolator choice and enabling a seamless, user-independent mitigation procedure. We further propose two algorithms that identify sets of isomorphic hardware layouts on which a given circuit yields sufficiently distinct expectation values to enable reliable zero-noise extrapolation. Experiments on a 133-qubit IBM Quantum device demonstrate that FF-ZNE yields mitigated expectation values with average deviations of approximately $6\%$ and $4.5\%$ for up to 50-qubit EfficientSU2 (sparse) and Hamiltonian simulation (dense) circuits, respectively, thus showing that this method is scalable, and is applicable for a broad range of circuits. By eliminating noise-factor and extrapolator selection, FF-ZNE transforms zero-noise extrapolation from a technique requiring expert tuning into a practical, scalable, and broadly accessible error-mitigation method for current quantum hardware.

\end{abstract}

\maketitle

\section{Introduction}
\label{sec:intro}

Noise remains the primary obstacle to extracting reliable results from present-day quantum computers as fully fault-tolerant architectures remain beyond current technological reach. As a result, most near-term quantum applications must contend with hardware errors directly, often limiting the achievable accuracy even for relatively shallow circuits. To address this challenge, \textit{Quantum Error Mitigation (QEM)} techniques aim to estimate ideal, noiseless expectation values without the overhead of full error correction, instead trading increased quantum and classical resources for improved accuracy. Various QEM methods have been studied in the literature ~\cite{Temme:2017p180509,  van2023probabilistic, kim2023evidence, virtual, van2022model} and have been shown to provide accurate estimates of noise-free expectation values under differing noise models and resource requirements.

Among these methods, \textit{Zero-Noise Extrapolation (ZNE)}~\cite{Temme:2017p180509} and \textit{Probabilistic Error Amplification (PEA)}~\cite{kim2023evidence} are particularly attractive in near-term settings due to their conceptual simplicity, and relatively low overhead. Both techniques operate by deliberately amplifying the effective noise in the computation and subsequently performing a classical extrapolation to infer the expectation value at the zero-noise limit. In ZNE, noise amplification is typically achieved by increasing the depth of the circuit through gate insertions (called folding) (see Sec.~\ref{sec:bg}), while in PEA the noise is amplified by explicitly injecting learned noise channels after each circuit layer. By learning the noise, and amplifying it, PEA can provide a rigorous error bound (apart from the extrapolation bias), which ZNE with gate folding fails to provide. Nevertheless, both methods can be experimentally validated and have seen widespread adoption.

Despite their apparent simplicity, the practical performance of ZNE and related extrapolation-based techniques is sensitive to methodological choices. In particular, the selection of noise amplification factors and the choice of extrapolator can lead to different mitigation outcomes, even for the same circuit and observable. Table~\ref{tab:motivation} illustrates a representative scenario in which different combinations of noise factors and extrapolator result in substantially different expectation values (\textit{expval}) at the zero-noise limit for the same circuit family.

\begin{table*}[t]
\small
\setlength{\tabcolsep}{3pt} 
\def\thickhline{\noalign{\hrule height1pt}} 
\caption{The mitigated expectation value for different $n$ (number of qubits) and $d$ (2-qubit depth) in a mirrored brickwork circuit~\cite{belkin2024approximate}, with observable 
$O = \frac{1}{n}\sum_i Z_i$ using standard ZNE~\cite{Temme:2017p180509} for different noise factors and extrapolators.}
\label{tab:motivation}
\centering

\begin{tabular}{C{3cm}cccr}
\hline
\textbf{Size of the circuit} & \textbf{Ideal expval} & 
\multicolumn{2}{c}{\bf ZNE parameters} &  \textbf{Mitigated expval} \\
 \textbf{($n \times d$)} &  & \em Noise factor & \em Extrapolator & \\
\hline
\hline
$8 \times 4$   & 1 & $(1,3,5)$       & linear       & 0.992 \\
$8 \times 4$   & 1 & $(1,3,5)$       & exponential  & 1.023 \\
\thickhline
$40 \times 40$ & 1 & $(1,3,5)$       & linear       & 0.581 \\
$40 \times 40$ & 1 & $(1,3,5)$       & exponential  & 23.815 \\
\hline
$40 \times 40$ & 1 & $(1,1.2,1.4)$   & linear       & 1.264 \\
$40 \times 40$ & 1 & $(1,1.2,1.4)$   & exponential  & $5.48 \times 10^{24}$ \\
\bottomrule
\end{tabular}

\end{table*}

In practice, there is no universally applicable or principled prescription for choosing the noise amplification factors and the extrapolation model in ZNE. Their selection, therefore, often relies on prior experience, heuristic guidelines \cite{majumdar, mitiq}, or empirical trial-and-error. In the absence of prior intuition, a naïve approach that tests parameter choices commonly reported in the literature may require executing on the order of $28$ distinct circuit instances solely to identify a suitable configuration (see Table~\ref{tab:nf}).

In this work, we propose Folding-Free Zero-Noise Extrapolation (FF-ZNE), an alternative approach to ZNE that avoids circuit folding altogether. Instead of amplifying noise by increasing circuit depth, FF-ZNE exploits the intrinsic non-uniformity of noise across current quantum devices by executing the same circuit on different isomorphic hardware layouts with distinct noise profiles. These executions yield expectation values with increasing deviation from the ideal as the effective noise of the layout increases, thereby emulating the noise-amplification mechanism of standard ZNE without modifying the circuit structure. 
We analytically show that under a depolarizing noise model, the appropriate extrapolator for FF-ZNE is fixed to linear. While the noise in practical quantum devices is not purely depolarizing, it can be converted to an effective stochastic Pauli channel via Pauli twirling~\cite{wallman2016noise}. We further numerically demonstrate that linear extrapolation remains effective under this twirled noise model, providing a principled justification for extrapolator selection in FF-ZNE.

We evaluate FF-ZNE on a 133-qubit IBM Quantum device using EfficientSU2 (sparse) and Hamiltonian simulation (dense) circuits of up to 50 qubits. Across these benchmarks, FF-ZNE achieves an average deviation of $< 6\%$ from the ideal expectation values for 50-qubit circuits. 
FF-ZNE removes the need for heuristic parameter tuning or repeated trial-and-error to select noise amplification factors and extrapolation models, addressing a key practical limitation of existing ZNE-based approaches. Our results indicate that FF-ZNE is effective at scale, and is applicable across different circuit classes on current hardware.

The rest of this paper is organized as follows: Section~\ref{sec:bg} introduces quantum error mitigation and discusses the limitations of current literature; in Section~\ref{sec:layout-based-zne}, we present FF-ZNE and two algorithms for identifying high-quality isomorphic circuit layouts to execute the circuit; Section~\ref{sec:result} reports experimental results for sparse EfficientSU2 circuits using Qiskit’s default layout noise profiling method~\cite{nation2023suppressing}; in Section~\ref{sec:qic} we propose an alternative noise profiling approach based on Quality Indicator Circuits~\cite{srivastava2025lightweighttargetedestimationlayout} that overcomes limitations of our default approach for dense circuits; Section~\ref{sec:exp_results_ham_circuit} reports experimental results for dense Hamiltonian simulation circuits using noise profiling approach based on Quality Indicator Circuits and we conclude in Section~\ref{sec:conclusion}.

\section{Related Work and Motivation}
\label{sec:bg}

\subsection{Quantum Error Mitigation}

Current quantum computers operate in the presence of hardware noise, which significantly limits computational accuracy. Although quantum error correction offers a long-term path toward fault-tolerant computation, present-day devices lack the qubit counts, coherence times and gate fidelities required to implement it at scale. As a result, \textit{Quantum Error Mitigation (QEM)} has emerged as a practical framework for improving the reliability of near-term quantum computing without the overhead of full error correction.
In general, QEM techniques aim to reduce the impact of noise on measured expectation values by modifying the way circuits are executed or the way measurements are processed. These typically perform additional circuit executions, calibration steps or classical post-processing to achieve better accuracy. Importantly, different QEM strategies target distinct error sources such as dephasing, gate errors or readout errors, and offer different trade-offs in effectiveness, overhead and hardware dependence.

Dynamical decoupling (DD) \cite{viola1999dynamical} suppresses idle‑time dephasing through well‑timed pulse sequences and remains valuable for protecting qubits during idle times. Its effectiveness depends on hardware‑specific pulse calibrations, and the duration of the idle time. Pauli twirling \cite{wallman2016noise} is similarly useful to convert the overall noise model of the system into stochastic Pauli noise by randomly inserting Pauli gates while maintaining functional equivalence. While this makes the noise easier to model and simulate, it does not reduce the underlying noise strength and can even increase the effective error rate due to the insertion of Pauli gates. Thus, while DD and twirling serve important roles in noise management and characterization, they do not directly address gate‑level noise.

Measurement errors are mitigated using methods such as \textit{M3}~\cite{nation2021scalable} and \textit{Twirled Readout Error Extinction (TREX)}~\cite{van2022model}. M3 models measurement errors using a reduced assignment matrix over observed outcomes, allowing numerical and scalable inversion of the assignment matrix to mitigate noise. TREX applies random Pauli-X twirling followed by classical bit flips to diagonalize the readout error matrix, enabling straightforward inversion after modest calibration. Both improve readout accuracy, differing in the approach and the classical postprocessing overhead, but do not address gate-level noise.

\textit{Probabilistic Error Cancellation (PEC)}~\cite{chen2023error} reconstructs the ideal circuit as a weighted linear combination of noisy circuits, producing unbiased estimates for twirled quantum circuits. However, the sampling overhead, and hence the number of circuit executions required, of PEC grows exponentially with circuit depth and noise strength, making PEC impractical for large or deep circuits despite its theoretical accuracy advantage.

Finally, \textit{Zero-noise extrapolation (ZNE)} estimates ideal expectation values by deliberately amplifying circuit noise, commonly through gate folding \cite{Temme:2017p180509} or deliberately injecting noise after learning the noise model for each circuit layer \cite{kim2023evidence}, and extrapolating the expectation value to the zero-noise limit using a chosen functional model. ZNE directly addresses gate-level noise during circuit execution, which often dominates in moderately deep circuits, and unlike PEC does not lead to an exponential number of circuit executions. ZNE, together with DD, Pauli twirling and measurement error mitigation, has been shown to produce significantly good mitigation in multiple studies \cite{majumdar, kim2023evidence}. Nevertheless, ZNE is a biased error mitigation technique (discussed in the following subsection) that, unlike PEC, does not provide a rigorous error bound on the mitigated expectation value.

In the next subsection, we focus on ZNE, examining its working principle and practical limitations.

\subsection{Zero-Noise Extrapolation (ZNE)}
Zero-noise extrapolation (ZNE) is a widely used QEM techniques due to its conceptual simplicity and relatively low quantum overhead. Numerous studies have demonstrated its practical effectiveness and investigated best practices for gate folding strategies~\cite{schultz2022impact}, extrapolation methods~\cite{giurgica2020digital} and hybrid approaches that combine ZNE with other mitigation techniques~\cite{majumdar}.

In ZNE, noise amplification is achieved by repeating quantum gates in a manner that preserves the ideal circuit functionality. For instance, a unitary gate $\hat{U}$ can be folded as $\hat{U}\hat{U}^\dagger\hat{U}$, where $\hat{U}^\dagger\hat{U}$ (equivalently $\hat{U}\hat{U}^\dagger$) implements the identity operation. This results in an effective threefold amplification of gate-level noise, without losing functional equivalence. This is commonly denoted by a noise-scaling factor $\lambda = 3$, with $\lambda = 1$ corresponding to the original circuit. By executing multiple functionally equivalent circuits at different values of $\lambda$, expectation values are obtained as a function of the noise scale and subsequently extrapolated to the zero-noise limit, yielding an estimate of the ideal expectation value (see Fig.~\ref{fig:zne_extrapolation}).

\begin{figure}[t]
    \centering
    \includegraphics[width=0.7\linewidth]{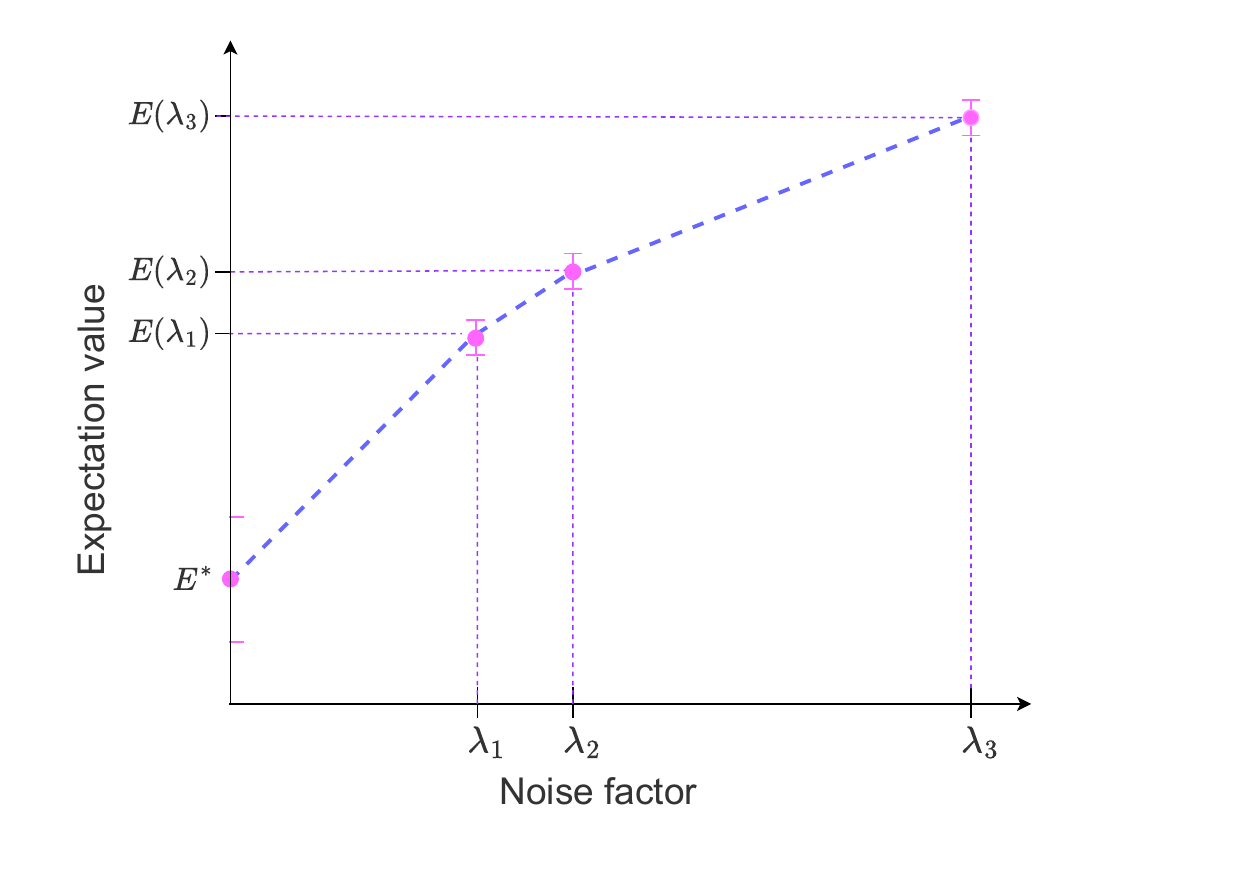}
    \caption{Illustration of ZNE technique, adapted from~\cite{majumdar}. 
    The expectation values $E(\lambda_{1}), E(\lambda_{2})$ and $E(\lambda_{3})$ are measured at increasing noise factors $\lambda_{1} < \lambda_{2} < \lambda_{3}$, and the zero-noise value $E^{*}$ is estimated by extrapolating these results to the zero-noise limit.}
    \label{fig:zne_extrapolation}
    
\end{figure}

A key assumption underlying ZNE is that increasing $\lambda$ scales the strength of noise without altering its qualitative form. In real quantum hardware, however, this assumption is only approximately satisfied. Different error sources, such as decoherence, gate error and control imperfections, exhibit distinct scaling behaviors under gate folding, which can degrade extrapolation accuracy. This issue is addressed using Pauli twirling~\cite{wallman2016noise}, which transforms hardware-specific noise into an effective stochastic Pauli channel with more uniform scaling properties. As demonstrated in~\cite{majumdar}, the combination of Pauli twirling with ZNE can substantially improve mitigation performance.

It is important to note that ZNE is fundamentally a \textit{biased} mitigation technique~\cite{mohammadipour2025direct}, where the model yields an estimate $\hat{E} = E^* + b$, $E^*$ being the ideal expectation value. The bias arises from the fundamental difference in the noise model after amplification, the imperfection of expectation value calculation due to finite sampling, and the uncertainty of extrapolation. A major drawback of ZNE is that it cannot provide a rigorous bound on the bias $b$.

\subsection{Limitations of Standard ZNE}
\label{limitation_zne}

ZNE faces two distinct classes of limitations: a fundamental theoretical bias, which defines the mathematical limitations of ZNE, and a practical challenge regarding parameter selection, which often makes the use of ZNE difficult for large-scale use-cases. In this work, we specifically target the second challenge i.e. addressing the bottleneck of parameter selection by providing an algorithmic framework that eliminates the need for ad-hoc noise factor selection, thus making ZNE more easily accessible to users in general. For error mitigation techniques with rigorous error bounds (e.g., necessary for quantum advantage \cite{lanes2025frameworkquantumadvantage}) should look into PEC \cite{chen2023error}, and its variants such as shaded lightcone \cite{qdc2025_slc_pec} and propagated noise absorption \cite{qdc2025_pna}.

\textbf{The parameter-selection problem:} A critical practical limitation of ZNE lies in the selection of noise-scaling factors and extrapolation models. Although ZNE is formally independent of circuit size, its empirical performance is sensitive to these choices. As illustrated in Table~\ref{tab:motivation}, two circuits belonging to the same family  but differing in qubit count and depth can exhibit markedly different ZNE performance when evaluated using identical noise factors and extrapolators. E.g., the noise factors $(1,3,5)$, which yield accurate results for an $8 \times 4$ circuit using both linear and exponential extrapolation, perform poorly for a $40 \times 40$ circuit. This indicates that ZNE is not self-calibrating, and suitable noise factors and extrapolators cannot be transferred reliably across circuit sizes or depths.

Several approaches have been proposed to address this parameter-selection problem, but all rely on some form of surrogate calibration or exhaustive search. \textit{Mitiq}~\cite{mitiq} evaluates multiple noise factors and extrapolators using Randomized Benchmarking (RB) circuits and selects the combination that performs best. However, RB circuits differ structurally from application circuits, and noise-scaling behavior observed under RB does not necessarily generalize to the circuit of interest. Another study has proposed constructing a database of empirically optimal noise factors and extrapolators indexed by circuit characteristics such as qubit count and two-qubit gate depth~\cite{majumdar}. While effective in limited settings, this requires extensive upfront characterization and ongoing maintenance for each circuit family, limiting its portability and scalability.

Alternative methods attempt to avoid explicit noise amplification by separating different noise sources or allowing noise to act naturally on the circuit~\cite{10313621}. While conceptually simple, these approaches become impractical for larger circuits and typically rely on circuit cutting~\cite{PhysRevLett.125.150504}, which introduces exponential overhead. A widely used practical strategy is to tune ZNE parameters by Cliffordizing the circuit of interest~\cite{qiskit_building_confidence_2024}, enabling direct comparison against classically tractable ideal values. While this avoids maintaining external databases and often yields good performance, it still requires exhaustively testing multiple candidate noise factors and extrapolators. Crucially, its success depends on the initial choice of noise factors to be tested; if this is poorly matched to the circuit, no amount of validation can recover accurate extrapolation.

\begin{table}[t]
\small
    \centering
    \caption{Common noise factors used in the literature for ZNE}
    \begin{tabular}{c|c|c}
    \hline
        \hline
        \bf Majumdar \etal~\cite{majumdar} & \bf Mitiq~\cite{mitiq_documentation} & \bf Kim \etal~\cite{kim2023evidence} \\
        \hline\hline
        [1, 3, 5], [1, 1.1, 1.2] & [1, 2, 3], [1, 1.5, 2, 2.5, 3],& [1, 1.2, 1.6]\\
        &[1, 3, 5, 7], [1.2, 1.4, 1.6, 1.8, 2.0], [1, 3] &\\
        \hline
    \end{tabular}
    \label{tab:nf}
\end{table}

The noise factors most commonly reported in literature for ZNE are listed in Table~\ref{tab:nf}. In the worst case, Cliffordization-based tuning requires evaluating all such candidates, resulting in 28 distinct circuit executions. On current quantum hardware, this corresponds to approximately $6.5$~minutes of execution time considering $10,000$ shots, making parameter tuning the dominant overhead of ZNE.  
Consequently, fundamental challenge in ZNE is that accurate mitigation requires prior knowledge of how noise scales with circuit execution, and yet this scaling behavior is itself circuit-dependent and generally unknown. Existing solutions to address this through heuristic selection, surrogate calibration or exhaustive trials do not scale favorably or need human intervention. This motivates the need for an alternative formulation of ZNE that avoids explicit noise-factor selection and extrapolator tuning altogether.

\subsection{Motivation}
In this work, we introduce \textit{Folding-Free Zero-Noise Extrapolation (FF-ZNE)}, which removes the need for gate folding and determines the extrapolation directly from hardware-executed circuits. Using a small number of isomorphic circuit layouts, FF-ZNE infers noise scaling without requiring prior information on suitable noise factors or extrapolation models. In our implementation, only three circuit executions are required. In the worst-case scenario, where no prior information is available and conventional ZNE relies on exhaustive trial-and-error over commonly used noise factors, this corresponds to a reduction in the number of circuit executions from $28$ (if all the standard noise factors commonly reported in the literature are tested) to $3$, providing an approximately $9\times$ decrease in quantum overhead (see Fig.~\ref{fig:runtime_comparison}). On real hardware, this translates to a decrease in total execution time from $6.5$ minutes to $42$ seconds. When partial prior knowledge of suitable noise factors is available, conventional ZNE may require fewer executions. However, FF-ZNE eliminates the need for such prior tuning entirely, offering a consistent and automated alternative.

\begin{figure}[t]
    \centering
    \includegraphics[width=0.6\textwidth]{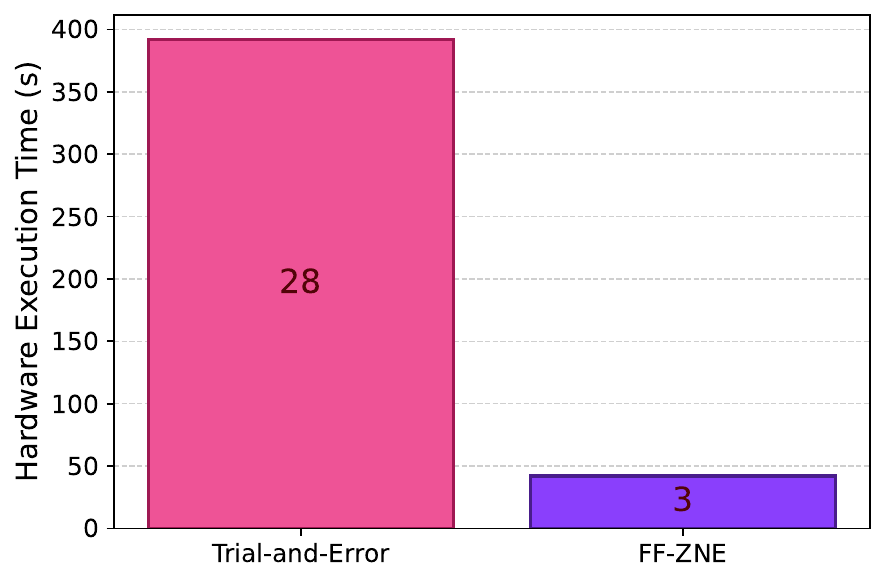}
    \caption{Comparison of the worst case quantum execution time for conventional ZNE and FF-ZNE. In the worst case, conventional ZNE requires exhaustive testing of standard noise-factor candidates (28 circuit executions), whereas FF-ZNE uses 3 isomorphic circuit executions, yielding an approximately $9\times$ reduction. On the hardware, this corresponds to a decrease in total execution time from $6.5$ minutes to $42$ seconds; conventional ZNE may require fewer executions when suitable noise factors are known \textit{a priori}.}
    \label{fig:runtime_comparison}
\end{figure}

\section{Folding-Free ZNE (FF-ZNE) Technique}
\label{sec:layout-based-zne}

In this section, we introduce FF-ZNE. We discuss two algorithms for the selection of layouts on which to execute the circuit, and also show that a \textit{linear} extrapolator is sufficient for FF-ZNE.

\subsection{Layout dependent impact of noise}
\label{sec:diff-layouts-diff-alg}

A \textit{layout} specifies the assignment of virtual qubits to physical qubits on the target device and is a central step in quantum circuit \textit{transpilation}. Layout selection is typically guided by SWAP minimization, an NP-hard problem addressed through heuristics~\cite{minimizeSWAP,sabre,zou2024lightsabre,tket}. Consequently, multiple circuit-structure–preserving isomorphic layouts naturally arise, which can nevertheless exhibit markedly different effective noise characteristics. We refer to the aggregate noise associated with a given layout as its \textit{noise profile}. Following the prior approach~\cite{PRXQuantum.4.010327}, we adopt a two-step strategy: (1) An initial layout is obtained to minimize SWAP overhead without explicit consideration of noise; and then (2) All layouts isomorphic to the resulting mapping are generated and scored according to their noise profile. This enables access to a set of layouts that preserve circuit structure while spanning a broad range of noise profiles.

Several methods for estimating layout-dependent noise have been proposed, differing in both the underlying noise models and the interpretation of their scores~\cite{PRXQuantum.4.010327, srivastava2025lightweighttargetedestimationlayout, wilson2020just, US20250165836A1}. In this work, we employ \textit{Mapomatic}~\cite{PRXQuantum.4.010327}, the default layout scoring method in Qiskit~\cite{javadi2024quantum}, and \textit{Quality Indicator Circuits (QIC)}~\cite{srivastava2025lightweighttargetedestimationlayout}. Both approaches provide lightweight, hardware-efficient estimates of layout quality and enable systematic ranking of isomorphic layouts according to their noise profiles.

\subsection{FF-ZNE with layout-induced noise diversity}
\label{sec:diff-layouts-diff-alg_part_b}

In conventional ZNE, the zero-noise expectation value of an observable is inferred by extrapolating measurements obtained from noise-amplified variants of the target circuit, typically generated through gate folding. In contrast, our approach exploits \textit{spatial noise inhomogeneity} intrinsic to current quantum processors. Specifically, we execute a fixed circuit on $k$ isomorphic layouts that preserve circuit structure while inducing different effective noise profiles. This gives a set of \textit{expectation values} whose deviations from the ideal correlate with the noise characteristics of the underlying layouts, thereby providing an alternative means of sampling the noise response of the circuit without modifying its gate structure. By \textit{extrapolating} these measurements to the zero-noise limit, we obtain an estimate of the noiseless expectation value. Crucially, this strategy avoids the need to explicitly construct or calibrate circuit-dependent noise amplification factors, which are often difficult to identify reliably in practice.

\subsection{Selection of the extrapolator}
\label{lin_extrapolator}
Next, we show that under the consideration of a \textit{global depolarizing noise}, the extrapolator is always \textit{linear} when the noise profiles of the selected isomorphic layouts are multiplicative constants of each other. 

For an $n$-qubit density matrix $\rho_l$ corresponding to a circuit executed on layout $l$, the effect of depolarizing noise can be represented as:
\[
\rho_l = (1 - p_l)\rho_{\text{in}} + \frac{p_l I}{2^n}.
\]
Therefore, the expectation value of an observable $O$ becomes
\[
\langle O\rangle_l = (1 - p_l)\langle O\rangle_{\text{ideal}} + p_l \langle O\rangle_{\text{noisy}}.
\]
If we consider two different layouts $l_1$ and $l_2$, and their probability of errors $p_{l_1}$ and $p_{l_2}$, where $p_{l_2} = \delta\, p_{l_1}$, then the corresponding expectation values for the execution on the two layouts are:
\begin{equation}
    \label{eq:1}
    \langle O\rangle_{l_1} = (1 - p_{l_1})\langle O\rangle_{\text{ideal}} + p_{l_1}\langle O\rangle_{\text{noisy}}
\end{equation}
\begin{equation}
    \label{eq:2}
    \langle O\rangle_{l_2} = (1 - \delta p_{l_1})\langle O\rangle_{\text{ideal}} + \delta p_{l_1}\langle O\rangle_{\text{noisy}}
\end{equation}
From Eq.~\ref{eq:1} and ~\ref{eq:2}, we obtain:
\begin{equation}
    \label{eq:3}
    \langle O\rangle_{\text{ideal}}
   = \frac{\delta}{\delta - 1}\,\langle O\rangle_{l_1}
   - \frac{1}{\delta - 1}\,\langle O\rangle_{l_2}.
\end{equation}
Eq.~\ref{eq:3} corresponds to \textit{Richardson extrapolation}~\cite{richardson1911ix, richardson1927viii}, which is linear in nature.

\begin{figure}[t]
    \centering
    \includegraphics[width=0.7\linewidth]{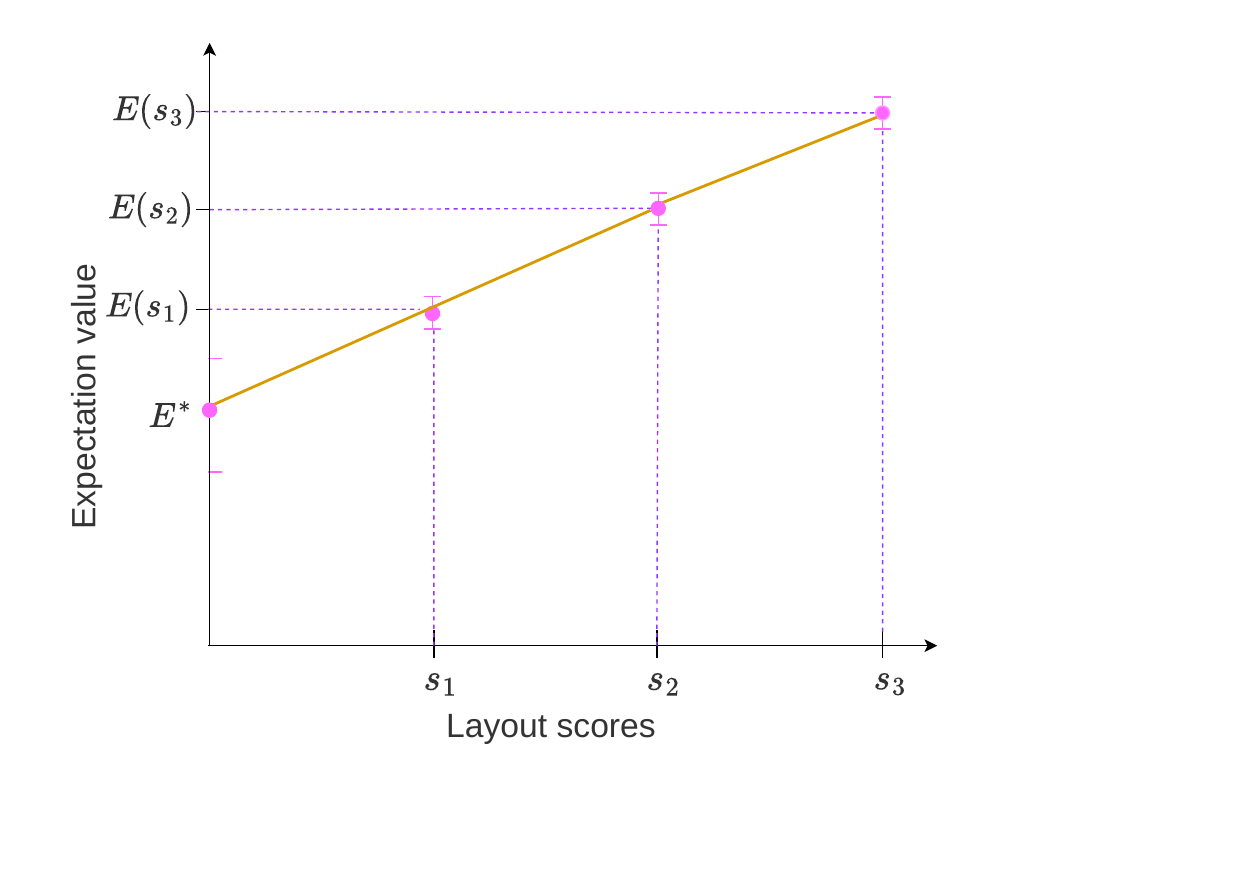}
    \caption{Illustration of FF-ZNE technique, with x-axis as the layout scores. The expectation values $E(s_{1}), E(s_{2})$ and $E(s_{3})$ are measured at increasing noise factors $s_{1} < s_{2} < s_{3}$, which are layout scores in this case, and the zero-noise value $E^{*}$ is estimated by extrapolating these results to the zero-noise limit.}
    \label{fig:zne_scores_exp}
    
\end{figure}

In practice, the noise affecting the device cannot be accurately described by a global depolarizing channel. Instead, we apply Pauli twirling to the native noise, which maps the underlying noise channel to an effective stochastic Pauli channel. In Sec.~\ref{sec:result}, we experimentally demonstrate that linear extrapolation remains effective under this stochastic Pauli noise model. Moreover, the physical error probability $p_l$ associated with a given layout $l$ is not directly accessible. Instead, layout-selection techniques such as Mapomatic and QIC provide a heuristic score $s_l = f(p_l)$, where $f$ is a generally monotonic, layout-dependent function of the underlying error probability. Consequently, recovering a noise-free estimate requires performing linear extrapolation directly with this heuristic score, rather than on the unobservable physical noise parameters as shown in Fig~\ref{fig:zne_scores_exp}.

\subsection{Selection of the isormophic layouts}
\label{subsec:selection}

The selection of the isomorphic layouts is critical for the success of FF-ZNE. If we select isomorphic layouts with a very similar noise profile, then the expectation values may not be separated enough to allow a good extrapolation. In addition, the choice of the isomorphic layouts $l_1, l_2, ..., l_k$, on which the circuit is to be executed, should be selected such that $|s_{l_{i+1}} - s_{l_{i}}|$ is roughly equal for all $i$ so that it corresponds to a linear extrapolator. Although two expectation values are sufficient to extrapolate using a linear curve, the usual practice of standard ZNE is to use three noise factors to better navigate the noise spectrum \cite{Temme:2017p180509, majumdar}; accordingly, we select three layouts for our approach. 
Finally, the layouts should also cover a significant portion of all isomorphic layouts. This ensures that the spectrum of noise variation of the hardware is properly explored, and not only a subset. To incorporate these requirements, we propose two algorithms to find the three desired layouts.

\subsubsection{Exhaustive Search}

Algorithm~\ref{alg:symmetric_layout_triples} performs an exhaustive search to find two layouts $(l_i, l_j)$ such that $\Delta =||s_{l_{i}} - s_{l_{j}}| - |s_{l_{1}} - s_{l_{i}}||$ is minimum. We always select the best layout with the minimum noise profile ($l_1$), and find the other two layouts by testing for every layout pair $(l_i, l_j)$, $i, j \neq 1$, $i < j$, to select the pair that gives the required condition stated above.
The algorithm starts with a list of circuit layouts sorted by their noise scores and fixes the layout $l_1$ with the lowest noise score. It then scans all possible pairs of the remaining layouts $(l_i, l_j)$, $1 < l_i < l_j$, such that $\Delta$ is minimum. Note that $\Delta$ may turn out to be minimum for very close-by layouts as well, which does not cover the full noise spectrum of the hardware. Therefore, we consider the weighted average of both $\Delta$ and the range of the layouts considered, controlled by the \textit{trade-off} parameter $a \in [0,1]$. It maintains a balance between selecting layouts that are far apart in the sorted list of isomorphic layouts and selecting layouts with low $\Delta$. A higher value of $a$ gives more importance to cover a wide spectrum of layouts, while a lower value gives more importance to the closeness of the score differences. Since the number of isomorphic layouts (integer, often going upto thousands or more) and the noise profile of a layout (between 0 and 1) have very different ranges, the Algorithm normalizes both of these values before calculating the weighted average. In this work, we choose $a=0.1$, which prioritizes minimizing $\Delta$ while still ensuring some coverage of the spectrum of layouts.

\begin{algorithm}[t!]
\caption{Exhaustive Search to identify best layout pair}
\label{alg:symmetric_layout_triples}
\small
\begin{algorithmic}[1]
    \REQUIRE
        Sorted list of $(\text{score}, \text{layout})$ pairs: $\text{scores\_sort} = \{(s_{l_{k}}, l_k)\}_{k=1}^m$; 
        trade-off parameter $a \in [0,1]$.
    \ENSURE
        The best double $(s_{l_{i}}, s_{l_{j}})$ minimizing the normalized weighted cost function, $s_{l_{1}}$ being the best layout score.

    \STATE Extract score list $\mathbf{s} \leftarrow [s_{l_{1}}, s_{l_{2}}, \ldots, s_{l_{m}}]$ and layout list $\mathbf{L} \leftarrow [l_1, l_2, \ldots, l_m]$
    \STATE Initialize empty list $\mathcal{P} \leftarrow []$
    \FOR{$i = 1$ to $m-1$}
        \STATE $\delta_1 \leftarrow |s_{l_{1}} - s_{l_{i}}|$
        \FOR{$j = i+1$ to $m$}
            \STATE $\delta_{\text{raw}} \leftarrow \big|\; \delta_1 - |s_i - s_j|\; \big|$
            \STATE Append $(i, j, \delta_{\text{raw}})$ to $\mathcal{P}$
        \ENDFOR
    \ENDFOR

    \STATE Extract vectors $i_{\text{vals}}, j_{\text{vals}}, \delta_{\text{vals}}$ from $\mathcal{P}$
    \STATE Normalize indices: 
    \[
        j_{\text{norm}} \leftarrow \frac{j_{\text{vals}}}{m-1}, 
        \quad (1-j_{\text{norm}}) \text{ computed element-wise.}
    \]

    \STATE Perform min–max normalization of $\delta_{\text{vals}}$:
    \[
\begin{aligned}
\delta_{\min} &= \min(\delta_{\text{vals}}), \delta_{\max} = \max(\delta_{\text{vals}}), \\[6pt]
\delta_{\text{norm}} &\leftarrow
\begin{cases}
\dfrac{\delta_{\text{vals}} - \delta_{\min}}{\delta_{\max} - \delta_{\min}}, & \text{if } \delta_{\max} > \delta_{\min}, \\[6pt]
0, & \text{otherwise.}
\end{cases}
\end{aligned}
\]

    \STATE Compute weighted cost for each pair:
    \[
        \text{cost} = a \cdot (1 - j_{\text{norm}}) + (1 - a) \cdot \delta_{\text{norm}}
    \]

    \STATE Sort all pairs $(i, j)$ by ascending cost and select the lowest-cost double.
    \STATE For each selected $( s_{l_{i}}, s_{l_{j}})$, report:
        \begin{itemize}
            \item raw scores $(s_{l_{1}}, s_{l_{i}}, s_{l_{j}}), s_{l_{1}}$ being the best layout score
            \item $\delta_{\text{raw}}$, $\delta_{\text{norm}}$, and $j_{\text{norm}}$
            \item corresponding layouts $(l_1, l_i, l_j), l_1$ being the best layout
        \end{itemize}
    \STATE \textbf{return} lowest double $(s_{l_{i}}, s_{l_{j}})$ with their layouts and cost metrics
\end{algorithmic}
\end{algorithm}

\begin{algorithm}[t!]
\caption{Binary Search for approximately equidistant midpoint}
\label{alg:binary_equidistant}
\small
\begin{algorithmic}[1]

\REQUIRE Sorted list of scores $\mathbf{s} \leftarrow [s_{l_{1}}, s_{l_{2}}, \ldots, s_{l_{m}}]$ and layout list $\mathbf{L} \leftarrow [l_1, l_2, \ldots, l_m]$, tolerance $\epsilon$
\ENSURE Index of score $s_{l_{i}}$ such that $\Delta$ is minimum.

\IF{$m < 3$}
    \RETURN \texttt{null} \COMMENT{Not enough elements}
\ENDIF

\STATE $s_{l_{1}} \gets l_1$, \quad $s_{l_{j}} \gets l_m$
\STATE $low \gets s_{l_{1}}$, \quad $high \gets s_{l_{j}}$
\STATE $best\_idx \gets low$, \quad $best\_diff \gets \infty$

\WHILE{$low \le high$}
    \STATE $mid \gets \lfloor (low + high)/2 \rfloor$
    \STATE $s_{l_{i}} \gets l_{mid}$
    \STATE $\delta_1 \gets s_{l_{i}} - s_{l_{1}}$, \quad $\delta_2 \gets s_{l_{j}} - s_{l_{i}}$
    \STATE $diff \gets |\delta_1 - \delta_2|$

    \IF{$diff < best\_diff$}
        \STATE $best\_diff \gets diff$
        \STATE $best\_idx \gets mid$
    \ENDIF

    \IF{$diff \le \epsilon$}
        \STATE \textbf{break} \COMMENT{Good enough, stop early}
    \ELSIF{$\delta_2 > \delta_1$}
        \STATE $low \gets mid + 1$ \COMMENT{Search right half}
    \ELSE
        \STATE $high \gets mid - 1$ \COMMENT{Search left half}
    \ENDIF

\ENDWHILE

\RETURN $best\_idx$

\end{algorithmic}
\end{algorithm}

\begin{lemma}
    Let $l_1$ be the minimum-noise layout among $m$ isomorphic layouts. The time complexity of Algorithm~\ref{alg:symmetric_layout_triples} for selecting layouts $l_i, l_j$, with $1 < i < j$, that minimize $\Delta$ is $\mathcal{O}(m^2)$.

\end{lemma}

\begin{proof}
Let $l_1$ denote the fixed layout with the minimum noise profile, and let $\{l_2, \dots, l_m\}$ denote the set of isomorphic layouts with associated scores $\{s_{l_2}, \dots, s_{l_m}\}$. Algorithm~\ref{alg:symmetric_layout_triples}
fixes the first layout $l_1$ and exhaustively searches for the remaining two layouts
$l_i$ and $l_j$, with $1 < i < j \leq m$, and finds the pair $(l_i, l_j)$ for which $\Delta$ is minimized. Each check requires a constant time. The algorithm thus scans over $(m-1)$ layouts for $l_i$ and $(m-2)$ layouts for $l_j$, thus leading to $(m-1)\cdot (m-2) = \mathcal{O}(m^2)$ checks.
\end{proof}

In Algorithm~\ref{alg:symmetric_layout_triples}, as we scale up the size of the target circuit, the number of isomorphic layouts increases, e.g., for a 40-qubit EfficientSU2 circuit on a 133 qubit quantum device the number of layouts $(m)$ was $1615432$. Algorithm~\ref{alg:symmetric_layout_triples} required more than 6 hours to find the three optimal layouts for this case. Therefore, we propose a binary search algorithm to calculate the desired layouts in lower time.

\subsubsection{Binary Search}
Here, we resort to a binary search to determine minimal $\Delta$ faster than Algorithm~\ref{alg:symmetric_layout_triples}. In Algorithm~\ref{alg:binary_equidistant}, we fix $l_1$ and $l_j$ as the first and last layouts, i.e., $j = m$, and perform a binary search to find $l_i$ such that $\Delta$ is minimized. The only difference from standard binary search is that, for some step of the algorithm, even when we search for $i$ between some arbitrary boundary $b_l$ and $b_r$, $\Delta$ is always calculated based on the original noise profile of $l_i$ and $l_j$. This ensures that the middle layout $l_i$ is decided based on the first and last layouts, and not on some arbitrary boundary. Throughout the process, the algorithm tracks the best candidate encountered and terminates early if a predefined tolerance is met. The output is the index of the layout that best approximates an equidistant midpoint between $l_1$ and $l_j$.

\begin{lemma}
    Given the first and last layouts $l_{1}$ and $l_{j}$, the time complexity of Algorithm~\ref{alg:binary_equidistant} for selecting a layout $l_{i}$, with $1 < i < j$, that minimizes $\Delta$ is $\mathcal{O}(m \log m)$, where $m$ is the number of isomorphic layouts.
\end{lemma}

\begin{proof}
Let $\mathcal{L}=\{l_1, l_2, \dots, l_m\}$ denote the isomorphic layouts, and $\{s_{l_1}, s_{l_2}, \dots, s_{l_m}\}$ be their associated scores sorted in ascending order. Consider the first and last layouts $l_1, l_j \in \mathcal{L}$ with $j=m$, which define an interval in the sorted list of layouts. The algorithm splits the search space in half in every iteration, and the search is continued in any of half depending on the noise profile score of the current midpoint layout $l_i$ -- which is essentially a binary search. Hence, the algorithm requires at most \(\mathcal{O}(\log m)\) evaluations to determine $i$ such that $\Delta$ is minimum, where \(m\) denotes the number of layouts.

\end{proof}

Note that Algorithm~\ref{alg:binary_equidistant} does not depend on on $l_1$ and $l_j$ being the first and last layouts in $\mathcal{L}$; they may correspond to any layout for which $1 < j$. If Algorithm~\ref{alg:binary_equidistant} is attempted $k$ times with $k$ different final layouts, then its runtime becomes $O(k\cdot m \log m)$, which is still $<<m^2$ when $k<<m$. In the worst case, if the last layout is tried with every layout, then this algorithm reduces to Algorithm~\ref{alg:symmetric_layout_triples}.

Although the binary search approach reduces the overall time complexity, it still requires evaluating, and scoring, the noise profile of all of the candidate isomorphic layouts. This becomes a time intensive and computationally expensive process when the number of layouts is large. In order to deal with this, we iteratively construct a set of accepted layouts using a qubit-overlap criterion. Let $L$ denote the set of layouts for which noise profile has been computed. For each candidate layout $l$ from the list of all isomorphic layouts, we compute its overlap with every layout $l_L \in L$. The noise profile of $l$ is computed only if $l \cap l_L < \eta$, $\forall$ $l_L \in L$ for some predefined $\eta$. The motivation behind this filtering is that two layouts with very high overlap are expected to have very similar noise profiles, and hence can be ignored. For our experiments, we have selected $\eta = 10$. This approximation reduced the number of layouts to be evaluated from $1,615,432$ to $12,768$ for a $40$-qubit EfficientSU2 circuit, thus reducing the overhead of noise profiling by $120\times$.

Since the final layout in Algorithm~\ref{alg:binary_equidistant} is fixed, unlike Algorithm~\ref{alg:symmetric_layout_triples} this \textit{does not
guarantee the optimal solution}. 
Moreover, the truncation of the layouts discussed earlier also dilutes the optimality of the method at the advantage of execution speed. For a 40-qubit EfficientSU2 circuit on a 133 qubit quantum device, Algorithm~\ref{alg:binary_equidistant} required only $2.2 \times 10^{-5}$ seconds to find the three optimal layouts. In Appendix~\ref{sec:appendix}, we show the value of $\Delta$ obtained by the two algorithms with and without truncation. We obtain a sufficiently good result for all these methods, with $\Delta$ ranging between $10^{-6}-10^{-4}$ for most of the cases.

\section{Experimental Results for EfficientSU2 circuit}
\label{sec:result}

In this section, we provide experimental results of FF-ZNE for $30$--$50$ qubits \emph{Cliffordized} EfficientSU2 circuit~\cite{dao2025exploring}. Cliffordized circuits allow us to obtain the noiseless expectation value for our circuits, which are mostly beyond brute force simulation, and hence validate the performance of our method. Fig.~\ref{fig:example_circuits_6q} shows 6-qubit EfficientSU2 circuit, and as an illustrative contrast, also shows a Hamiltonian simulation circuit. All experimental results in this section correspond to $30$--$50$ qubit circuits.

\begin{figure}[t]
    \centering
    \subfloat[A 6-qubit EfficientSU2 ansätze.
    ]{%
    \includegraphics[scale=0.365]{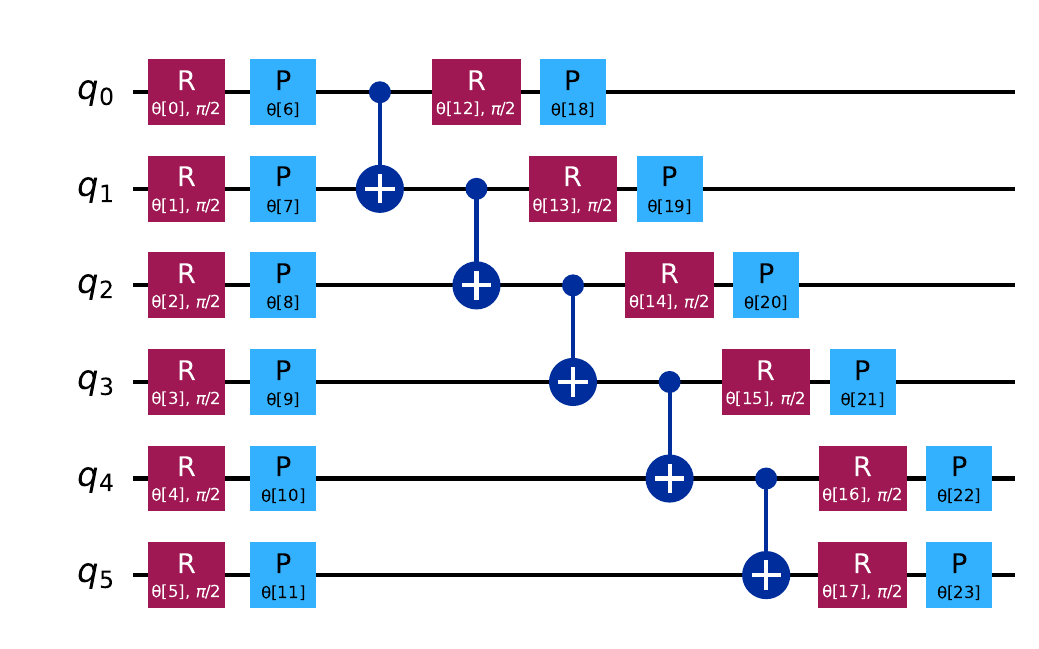}%
    \label{fig:effsu2_6q}%
    }%
    
    \subfloat[A 6-qubit Hamiltonian simulation circuit.]{%
    \includegraphics[width=0.52\textwidth]{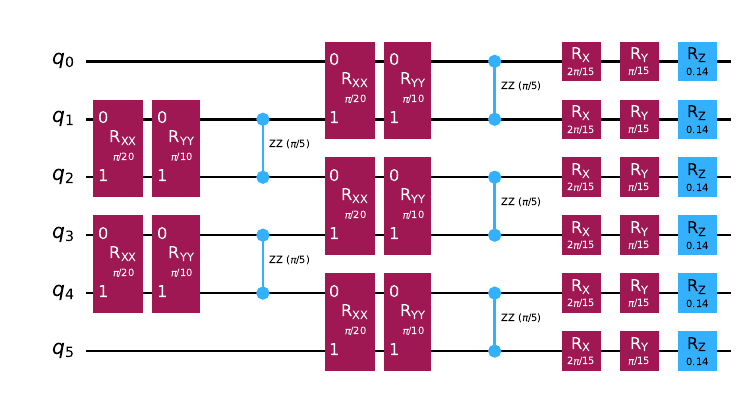}%
    \label{fig:ham_sim_6q}%
    }
    
    \caption{
    Example quantum circuits for $6$ qubits.
    (a) EfficientSU2 ansätze.
    (b) Hamiltonian simulation circuit.
    These are representative examples used for illustration.
    } 
    \label{fig:example_circuits_6q}
\end{figure}

\subsection{Noisy simulation results}
\label{results_noisy_sim}

We validate our method first using noisy simulations. Cliffordization and noisy simulation were performed with the Qiskit \emph{NEAT} tool~\cite{neat_tool}. For FF-ZNE, we require access to multiple isomorphic layouts. Since the Qiskit transpiler returns only the layout with the lowest estimated noise profile, we modify the \emph{VF2PostLayout} pass to return all isomorphic layouts along with their noise scores. We employ the default noise profiling technique of the Qiskit transpiler, \textit{Mapomatic}~\cite{nation2023suppressing}, where layout scores lie between $[0,1]$, with lower values indicating lower estimated noise. Layouts with a score of $1$ are excluded as such scores correspond to non-functional qubits or invalid connectivity. To avoid distortion from anomalously noisy qubits or couplings, we further discard layouts whose scores exceed the mean by more than three standard deviations.

Neither ZNE nor FF-ZNE addresses measurement noise. Moreover, Qiskit currently does not support measurement error mitigation (e.g., TREX) within noisy simulation. Accordingly, measurement errors were disabled in the simulation. The NEAT simulator also disables relaxation noise by default, as it is non-unitary. Under these conditions, the simulator implements an effective Pauli noise model, allowing the use of linear extrapolation as described in Sec.~\ref{lin_extrapolator}.

Using Algorithm~\ref{alg:binary_equidistant}, we then select three layouts $l_1, l_i, l_j$ such that $\Delta$ is minimum. The Cliffordized circuit is executed on each of them, and the expectation value of the observable $O = \frac{1}{n} \sum_{i=0}^{n-1} Z_i$, $n$ being the number of qubits, is calculated. Finally, we fit a linear curve for these three expectation values corresponding to the 3 layout scores, and extrapolate to the zero-noise limit.

\begin{figure}[t]
\begin{center}
\includegraphics[width=0.5\textwidth ]{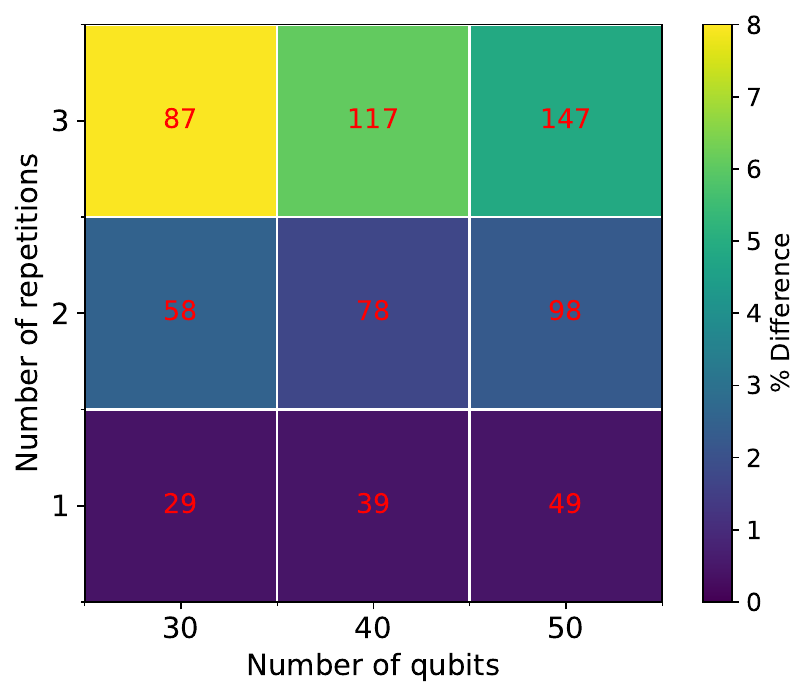}
\caption{A heatmap of the discrepancy between the ideal and the zero-noise expectation values derived using FF-ZNE. We test EfficientSU2 circuits with different numbers of qubits $n$ ranging from 30-50 (x-axis) and repetitions $r$ ranging from 1-3 (y-axis) using the observable $O = \frac{1}{n} \sum_{i=0}^{n-1} Z_i$. The value within each cell is the 2-qubit depth of the corresponding circuit.}
\label{fig:vqe_heatmap}
\end{center}

\end{figure}

Figure~\ref{fig:vqe_heatmap} shows a heatmap comparison between ideal and zero-noise expectation values obtained using our method for EfficientSU2 circuits. We consider system sizes ranging from $n=30$ to $50$ qubits and repetition layers $r=1$ to $3$. For EfficientSU2 ansätze, the two-qubit gate depth scales approximately linearly with both system size and circuit depth, $d_{2q} \approx n r$. In all configurations tested, the average relative deviation of the mitigated expectation values from the ideal values lies between $2.51\%$ and $3.63\%$, as opposed to the unmitigated results which lies between $3.03\%$ and $23.08\%$ away from ideal value. This demonstrates stable and consistent performance across different circuit depths and sizes, along with a clear improvement over the unmitigated results.

\subsection{Hardware execution results}
\label{results_hardware}

Next, we perform experiments on a 133-qubit IBM Quantum Heron processor using a 50-qubit EfficientSU2 circuit. The circuit was cliffordized to enable efficient classical evaluation of the ideal expectation value. Figure~\ref{fig:vqe_50qubit} compares the ideal expectation value with the corresponding zero-noise extrapolated estimates for circuits, with $r=1$ and $2$ repetitions, corresponding to two-qubit gate depths of $49$ and $98$, respectively.
We incorporate Pauli twirling to approximate the noise channel as a stochastic Pauli process, and TREX mitigation to reduce state-preparation and measurement (SPAM) errors. Since this circuit is sparse (see Fig.~\ref{fig:effsu2_6q}), we also incorporate Dynamical Decoupling (DD) using the XY4 pulse sequence. Using FF-ZNE, the mitigated expectation values deviate from the ideal values by an average of approximately $6\%$, as opposed to the unmitigated result which was $16.84\%$ away from ideal value. It shows that the proposed FF-ZNE technique significantly mitigates hardware noise and yields results close to ideal value on real quantum hardware.

\begin{figure}[t]
\centering
\includegraphics[width=0.7\textwidth ]{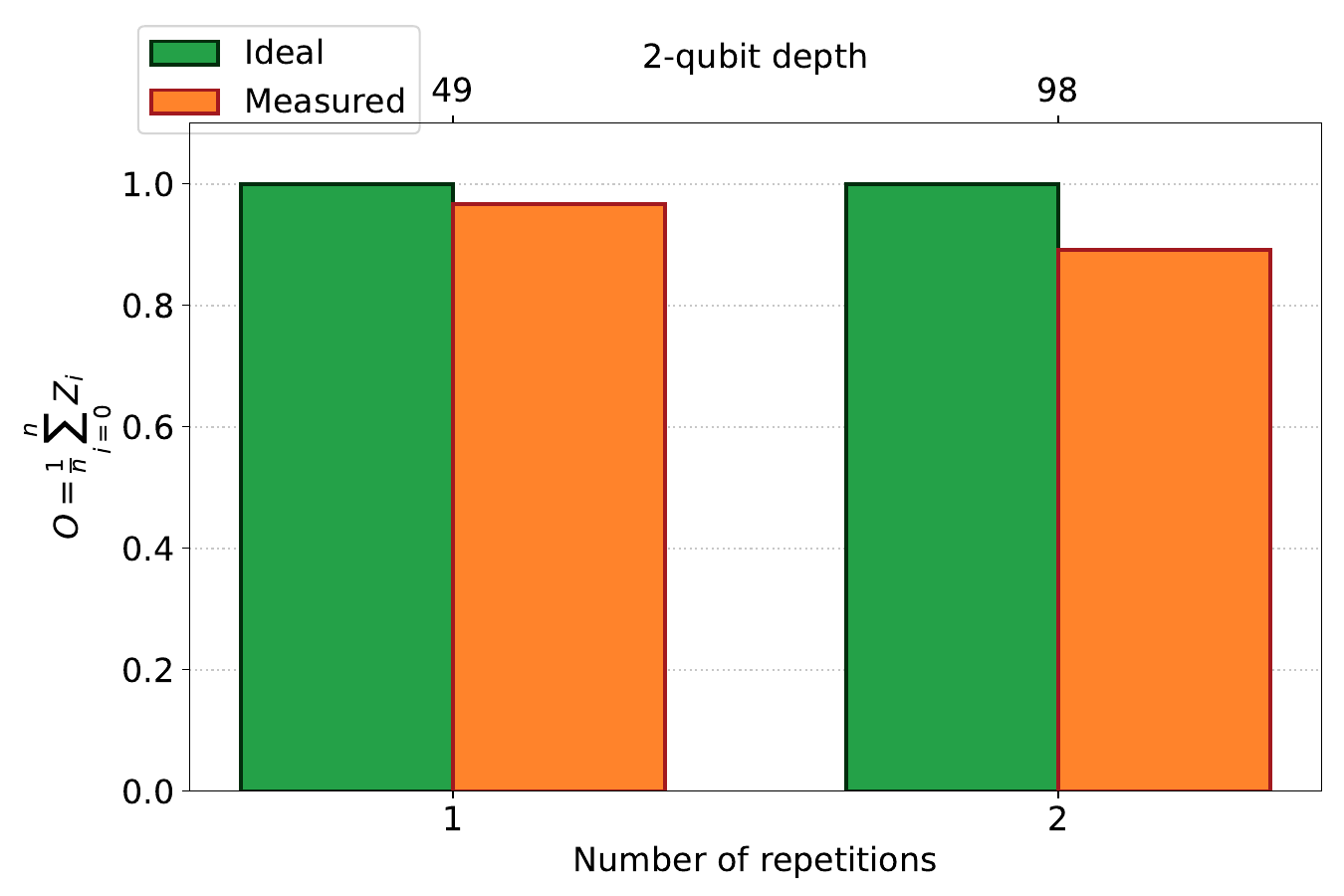}
\caption{Comparative analysis of the expectation values for a 50-qubit EfficientSU2 circuit with number of repetitions of the circuit (lower x-axis), 2-qubit depth of the circuit (upper x-axis) and the expectation value of the observable $O=\frac{1}{n} \sum_{i=0}^{n-1} Z_i$(y-axis)}
\label{fig:vqe_50qubit}

\end{figure}

\begin{table}[t]
\small
\centering
\caption{Expectation values of a 40-qubit EfficientSU2 circuit for observables of different weights.}
\label{tab:expvals_40q}

\begin{tabular}{c|c|c|c|c}
\hline
\multirow{2}{*}{\bf Repetitions} 
& \multirow{2}{*}{\bf Ideal Expval}
& \multicolumn{3}{c}{\bf Expectation Value} \\
\cline{3-5}
&  & \em Weight--1 & \em Weight--2 & \em Weight--3 \\
\hline\hline
1 & 1.0 & 1.0044 & 1.0054 & 1.0025 \\
2 & 1.0 & 1.0170 & 1.0263 & 1.0254 \\
3 & 1.0 & 1.0606 & 1.0734 & 1.0744 \\
\hline
\end{tabular}
\end{table}

\subsection{Results for higher-weight observables in VQE circuit}

A natural question is whether the weight of the measured observable impacts the performance of FF-ZNE. In principle, the observable weight should not directly affect FF-ZNE accuracy, since increasing the weight primarily amplifies SPAM contributions, which are mitigated independently via TREX. To test this expectation explicitly, we evaluate FF-ZNE on a 40-qubit EfficientSU2 circuit using observables of increasing weight. Specifically, we consider translationally averaged Pauli-Z strings of weight 1, 2 and 3, defined respectively as $\frac{1}{n} \sum_{i=0}^{n-1} Z_i$, $\frac{1}{n-1} \sum_{i=0}^{n-2} Z_iZ_{i+1}$ and $\frac{1}{n-2} \sum_{i=0}^{n-3} Z_iZ_{i+1}Z_{i+2}$. The resulting expectation values, reported in Table~\ref{tab:expvals_40q}, show that FF-ZNE yields comparable mitigation performance across all tested observable weights. This indicates that its effectiveness extends beyond single-qubit observables.

\section{QIC Scoring to Overcome Limitations for Dense Circuits}
\label{sec:qic} 

In this section, we examine the limitations of the default \textit{Mapomatic} scoring and demonstrate the use of Quality Indicator Circuit (QIC)~\cite{srivastava2025lightweighttargetedestimationlayout} scoring as an alternative.

\subsection{Limitations of the default profiling method}
\label{subsec:limitations_default_profiling_method}

Table \ref{tab:expvals} reports the zero-noise expectation value of the observable $O = \frac{1}{n} \sum_{i=0}^{n-1} Z_i$ obtained from a 20-qubit Hamiltonian simulation circuit executing on IBM Quantum hardware, with the number of Trotter steps varied between one and three. Across all cases, the mitigated expectation values exhibit a pronounced discrepancy from the corresponding ideal value, ranging from $1.62$--$2.27$ instead of the ideal $0.92$.

\begin{table}[t]
\small
\centering
\caption{Expectation values for the 20-qubit Hamiltonian simulation circuit at different trotter steps where the observable is $O = \frac{1}{n} \sum_{i=0}^{n-1} Z_i$}
\label{tab:expvals}
\begin{tabular}{c|c|c|c}
\hline
\textbf{Ideal Expval} & \textbf{Repetition 1} & \textbf{Repetition 2} & \textbf{Repetition 3} \\
\hline\hline
0.92 & 1.620 & 2.062 & 2.270 \\
\hline
\end{tabular}
\end{table}

This skew can be attributed to the limitations of the Mapomatic layout noise profiling, which follows an exponential decay model. For dense circuits, having a large number of gates, such as those used in Hamiltonian simulation, even relatively shallow circuits are assigned high layout scores. As shown in Table~\ref{tab:mapomatic_scores}, the separation between layout scores for Hamiltonian simulation circuits is substantially smaller than that observed for EfficientSU2 circuits. For example, in repetition 1, layout score differences for VQE circuits are on the order of $0.3$--$0.4$, whereas the corresponding difference for Hamiltonian simulation circuits is approximately $0.1$. This reduced dynamic range hampers reliable extrapolation, leading to degraded error-mitigation performance.

\begin{table}[t]
\small
\centering
\caption{Layout scores using Mapomatic algorithm for 20-qubit VQE and Hamiltonian simulation circuits across different repetitions.}
\label{tab:mapomatic_scores}

\begin{tabular}{c|c|c|c|c}
\hline
\textbf{Circuit} & \textbf{Reps} & \textbf{Layout $s_1$} & \textbf{Layout $s_i$} & \textbf{Layout $s_j$} \\
\hline\hline
\multirow{3}{*}{\bf VQE}
 & 1 & 0.204 & 0.602 & 0.999 \\
 & 2 & 0.343 & 0.671 & 0.999 \\
 & 3 & 0.496 & 0.748 & 0.999 \\
\hline
\multirow{3}{*}{\bf Hamiltonian}
 & 1 & 0.698 & 0.849 & 0.999 \\
 & 2 & 0.764 & 0.882 & 0.999 \\
 & 3 & 0.757 & 0.878 & 0.999 \\
\hline
\end{tabular}
\end{table}

This motivates the requirement of some other noise profiling method which decays slower than Mapomatic.

\subsection{Quality Indicator Circuits for noise profiling}
\label{quality_indicator_circuits_noise_profiling}

A \textit{Quality Indicator Circuit (QIC)}~\cite{srivastava2025lightweighttargetedestimationlayout} is a circuit whose ideal, noise-free outcome is known \textit{a priori} or can be computed efficiently on a classical computer. Importantly, the construction of a QIC is tailored to, and therefore dependent on, the structure of the target circuit of interest. The construction of a QIC satisfies three criteria: (i) the ideal outcome is known, (ii) the circuit preserves the essential structural features of the target circuit, and (iii) it can typically be implemented with a reduced circuit depth relative to the original circuit.

In the original paper~\cite{srivastava2025lightweighttargetedestimationlayout}, QIC is constructed as a network of CNOT gates acting on pairs of qubits, designed to mirror the two-qubit gate connectivity of the target circuit. This CNOT network is enclosed by two layers of Hadamard gates applied to all qubits, resulting in a non-trivial identity circuit with a known ideal outcome. When executed on a given layout, deviations from the ideal outcome provide an estimate of the corresponding noise profile. Prior results show that the decay rate obtained from QIC-based profiling is substantially slower than that observed using Mapomatic, motivating its use for noise profiling within FF-ZNE applied to Hamiltonian simulation circuits.

\section{Experimental Results for Hamiltonian Simulation Circuit}
\label{sec:exp_results_ham_circuit}

We report experimental results for Hamiltonian simulation circuits to assess the effectiveness of FF-ZNE when combined with the QIC method. 
Results for Hamiltonian simulation are presented exclusively on IBM Quantum hardware, as the validity of the proposed approach was established in the preceding section using both noisy simulations and hardware experiments. Accordingly, simulation results are omitted here, since experimental hardware data provide a more stringent test of the method.

Table \ref{tab:mitigated_expval_comparison_20_50} presents the mitigated expectation values for Hamiltonian simulation circuits with 20 and 50 qubits, where the circuits were cliffordized to enable computation of the ideal expectation values. Results obtained using both Mapomatic and QIC-based noise profiling are reported. As in the earlier experiments with EfficientSU2 circuits, we additionally employed Pauli twirling and TREX for the Hamiltonian simulation circuits. Owing to the high gate density of these circuits (see Fig.~\ref{fig:ham_sim_6q}), dynamical decoupling was not applied. 

\begin{table}[t]
\small
\centering
\caption{Comparison of QIC and Mapomatic mitigated expectation values for 20-qubit and 50-qubit Hamiltonian simulation circuits}
\label{tab:mitigated_expval_comparison_20_50}
\begin{tabular}{c|c|c|r|r}
\hline
\textbf{\# qubits} & \textbf{Reps} & \textbf{Ideal Expval} & \multicolumn{2}{c}{\textbf{Mitigated Expval}} \\
\cline{4-5}
 &  &  & \textit{QIC} & \textit{Mapomatic} \\
\hline\hline
\multirow{3}{*}{20} 
 & 1 & \multirow{3}{*}{1.0} & 1.039 & 1.620 \\
 & 2 & & 0.967 & 2.062 \\
 & 3 & & 0.962 & 2.270 \\
\hline
\multirow{3}{*}{50} 
 & 1 & \multirow{3}{*}{1.0} & 1.068 & 1.433 \\
 & 2 & & 1.062 & 29.318 \\
 & 3 & & 1.029 & 423.355 \\
\hline
\end{tabular}
\end{table}

Consistent with prior expectations, QIC-based scoring, characterized by a slower decay rate than Mapomatic, yields mitigated expectation values in closer agreement with the ideal results. Table~\ref{tab:mitigated_expval_comparison_20_50} shows the mitigated expectation value obtained for 20 and 50 qubits Hamiltonian simulation circuit, where the circuit was cliffordized to obtain the ideal expectation value. We show the mitigated expectation value using both Mapomatic and QIC noise profiling technique. As expected, QIC scoring, which has a slower decay rate than Mapomatic, shows close acceptance with the ideal expectation value. We find that the average deviation of the mitigated expectation value using QIC scoring from the ideal value is only $\approx 4.48\%$ as opposed to the unmitigated result which was $9.15\%$ away from ideal value.

\section{Conclusion}
\label{sec:conclusion}

In this paper, we have shown that the performance of the widely used ZNE error-mitigation technique is highly sensitive to the choice of noise-amplification factors and extrapolation models. This limits its practicality and robustness for large or diverse workloads. We propose FF-ZNE as an alternative that eliminates the need for explicit noise-scaling factors by leveraging isomorphic hardware layouts with distinct noise characteristics. This effectively emulates noise amplification by systematically profiling layouts according to their intrinsic noise, without modifying circuit structure. We obtain mitigated expectation values for EfficientSU2 circuits with up to 50 qubits that deviate from ideal values by only $\approx 6\%$.
We further address the limitations of using Mapomatic for scoring of dense circuits, such as Hamiltonian simulations, using QIC-based profiling. Its slower decay in noise characterization provides improved discrimination between layouts. This results in extrapolated expectation values for Hamiltonian simulation circuits with a deviation of only $\approx 4.5\%$ from ideal values on quantum hardware.

Overall, these results establish a scalable and automated alternative to conventional ZNE that lowers the barrier to effective error mitigation. By removing the need for manually selecting noise-amplification factors, the proposed method enables users without deep expertise in noise modeling to apply ZNE reliably, without incurring the substantial overhead of repeated circuit executions associated with trial-and-error parameter tuning. This broadens the practical usability of ZNE and facilitates more efficient execution of complex quantum workloads on near-term quantum devices. However, FF-ZNE, just like standard ZNE, does not provide a rigorous error bound, and retains the inherent bias, which cannot be bounded analytically. Our contribution is not in providing any theoretical bound for ZNE or FF-ZNE, but rather in offering a reliable algorithmic framework that avoids the guesswork currently required for ZNE implementations. For error mitigation techniques with mathematically rigorous error bounds, one should consider alternative strategies such as PEC~\cite{chen2023error}, SLC~\cite{qdc2025_slc_pec} or PNA~\cite{qdc2025_pna}.

\section*{Acknowledgment}
The authors acknowledge the contribution of Dr. Ritajit Majumdar from the Quantum Algorithm and Engineering team of IBM Quantum with detailed discussions in developing the problem statement and setting up the algorithmic and experimental strategies.

\subsection*{Code availability}
The codes used for the experiments reported in this manuscript are available at \url{https://github.com/dream-lab/quantum-FF-ZNE}.

\balance

\bibliographystyle{plain}
\bibliography{paper}

\clearpage
\appendix
\section{Appendix}
\label{sec:appendix}

Table~\ref{tab:delta_raw_clean} reports the values of $\Delta$ obtained for the EfficientSU2 circuit under different layout search strategies of FF-ZNE, highlighting the variation in $\Delta$ across different circuit sizes.

\begin{table*}[h!]
\centering
\caption{Values of $\delta$ for EfficientSU2 circuit under different layout search strategies.}
\label{tab:delta_raw_clean}
\setlength{\tabcolsep}{3pt} 
\renewcommand{\arraystretch}{0.9} 
\begin{tabular}{c|c|c|c|c|C{2cm}}
\hline
\textbf{Qubits} & \textbf{Algorithm} & \textbf{Repetitions} & $\delta$ & \textbf{Time (s)} & \textbf{\# of layouts} \\
\hline\hline
\multirow{12}{*}{10}
& \multirow{3}{*}{Bruteforce}
  & 1 & $1.56\times10^{-4}$ & 0.7 & 2770 \\
& & 2 & $1.42\times10^{-4}$ & 0.6 & 2770 \\
& & 3 & $5.43\times10^{-5}$ & 0.6 & 2770 \\
\cline{2-6}
& \multirow{3}{*}{Binary}
  & 1 & $1.56\times10^{-4}$ & $0.485\times10^{-3}$ & 2770 \\
& & 2 & $2.54\times10^{-3}$ & $0.279\times10^{-3}$ & 2770 \\
& & 3 & $6.24\times10^{-4}$ & $0.348\times10^{-3}$ & 2770 \\
\cline{2-6}
& \multirow{3}{*}{Truncated Binary}
  & 1 & $4.85\times10^{-2}$ & $0.449\times10^{-3}$ & 221 \\
& & 2 & $2.08\times10^{-2}$ & $0.213\times10^{-3}$ & 221 \\
& & 3 & $2.71\times10^{-2}$ & $0.213\times10^{-3}$ & 221 \\
\cline{2-6}
& \multirow{3}{*}{Truncated Bruteforce}
  & 1 & $1.25\times10^{-4}$ & $6.226\times10^{-3}$ & 221 \\
& & 2 & $3.58\times10^{-4}$ & $6.295\times10^{-3}$ & 221 \\
& & 3 & $6.61\times10^{-4}$ & $6.600\times10^{-3}$ & 221 \\
\hline
\multirow{12}{*}{15}
& \multirow{3}{*}{Bruteforce}
  & 1 & $5.12\times10^{-5}$ & 6.0 & 10354 \\
& & 2 & $5.57\times10^{-5}$ & 6.0 & 10354 \\
& & 3 & $6.27\times10^{-5}$ & 6.3 & 10354 \\
\cline{2-6}
& \multirow{3}{*}{Binary}
  & 1 & $5.16\times10^{-4}$ & $ 2.067\times10^{-3}$ & 10354 \\
& & 2 & $2.41\times10^{-4}$ & $ 4.002\times10^{-3}$ & 10354 \\
& & 3 & $5.30\times10^{-4}$ & $ 1.525\times10^{-3}$ & 10354 \\
\cline{2-6}
& \multirow{3}{*}{Truncated Binary}
  & 1 & $7.34\times10^{-4}$ &  $0.279\times10^{-3}$ & 779 \\
& & 2 & $1.90\times10^{-2}$ &  $0.236\times10^{-3}$ & 779 \\
& & 3 & $8.18\times10^{-4}$ &  $0.230\times10^{-3}$ & 779 \\
\cline{2-6}
& \multirow{3}{*}{Truncated Bruteforce}
  & 1 & $3.97\times10^{-4}$ &  0.04239 & 779 \\
& & 2 & $6.25\times10^{-6}$ &  0.04401 & 779 \\
& & 3 & $8.18\times10^{-4}$ &  0.04340 & 779 \\
\hline
\multirow{12}{*}{20}
& \multirow{3}{*}{Bruteforce}
  & 1 & $1.90\times10^{-5}$ & 39.8 & 33000 \\
& & 2 & $4.23\times10^{-6}$ & 39.6 & 33000 \\
& & 3 & $5.12\times10^{-6}$ & 39.3 & 33000 \\
\cline{2-6}
& \multirow{3}{*}{Binary}
  & 1 & $2.57\times10^{-4}$ & $1.161\times10^{-3}$ & 33000 \\
& & 2 & $1.32\times10^{-4}$ & $1.681\times10^{-3}$ & 33000 \\
& & 3 & $4.72\times10^{-5}$ & $2.789\times10^{-3}$ & 33000 \\
\cline{2-6}
& \multirow{3}{*}{Truncated Binary}
  & 1 & $5.37\times10^{-4}$ & $0.486\times10^{-3}$ & 2551 \\
& & 2 & $7.54\times10^{-4}$ & $0.494\times10^{-3}$ & 2551 \\
& & 3 & $3.17\times10^{-4}$ & $0.417\times10^{-3}$ & 2551 \\
\cline{2-6}
& \multirow{3}{*}{Truncated Bruteforce}
  & 1 & $5.63\times10^{-5}$ & 0.48946 & 2551 \\
& & 2 & $3.68\times10^{-5}$ & 0.29738 & 2551 \\
& & 3 & $2.08\times10^{-5}$ & 0.36019 & 2551 \\
\hline
\end{tabular}
\end{table*}

\end{document}